\documentclass[12pt]{amsart}

\usepackage{boxedminipage}
\usepackage{amsfonts}
\usepackage{amsmath} 
\usepackage{amssymb}
\usepackage{graphicx}
\usepackage{amsthm}
\usepackage{t1enc}
\usepackage{subfig}
\usepackage{tikz}
\usepackage{braket}
\usepackage{caption}
\usepackage{mathtools}

\usetikzlibrary{patterns}
\usetikzlibrary{arrows,shapes,positioning}
\usetikzlibrary{decorations.markings}
\usetikzlibrary[matrix] 
\tikzstyle arrowstyle=[scale=1]

\tikzstyle directed=[postaction={decorate,decoration={markings,
    mark=at position .65 with {\arrow[arrowstyle]{latex}}}}]

\newcommand{\lp}{\left(}
\newcommand{\rp}{\right)}

\newcommand{\tH}{\tilde{H}}

\newcommand{\R}{\mathbb{R}}
\newcommand{\N}{\mathbb{N}}
\newcommand{\Nz}{\N_0}
\newcommand{\Z}{\mathbb{Z}}
\newcommand{\Mz}{{\Z_-}}
\newcommand{\C}{\mathbb{C}}
\newcommand{\Q}{\mathbb{Q}}

\newcommand{\cM}{\mathcal{M}}
\newcommand{\cZ}{\mathcal{Z}}

\newcommand{\Klam}{K^{(\lambda)}}
\newcommand{\Jlam}{J^{(\lambda)}}

\newcommand{\Dlam}{D^{(\lambda)}}
\newcommand{\Rlam}{R^{(\lambda)}}

\newcommand{\KM}{K_M}
\newcommand{\KMpos}{{\KM^+}}
\newcommand{\KMneg}{{\KM^-}}
\newcommand{\KMn}{{\KM^{(n)}}}
\newcommand{\LMn}{{L_M^{(n)}}}
\newcommand{\LM}{L_M}

\newcommand{\hH}{\hat{H}}

\newcommand{\HM}{H_M}
\newcommand{\Hk}{\operatorname{Hk}}
\newcommand{\hk}{\operatorname{hk}}

\newcommand{\Slam}{S^{(\lambda)}}
\newcommand{\Psilam}{\Psi^{(\lambda)}}
\newcommand{\Philam}{\Phi^{(\lambda)}}

\newcommand{\Mlam}{M^{(\lambda)}}

\newcommand{\lspan}{\operatorname{span}}

\definecolor{ao(english)}{rgb}{0.0, 0.5, 0.0}

\makeatletter
\renewcommand{\boxed}[1]{\raisebox{1pt}{\text{\fboxsep=1pt\fbox{\m@th$\displaystyle#1$}}}}
\makeatother

\newcommand{\filledbox}{\tikz{\fill circle (3pt)}}
\newcommand{\emptybox}{\tikz{\draw circle (3pt)}}

\newcommand{\Wr}{\operatorname{Wr}}

\newcommand{\bt}{{\boldsymbol{t}}}

\newcommand{\flip}{f}

\newtheorem{theorem}{Theorem}

\newtheorem{proposition}[theorem]{Proposition}

\DeclareMathOperator{\bX}{\mathbf{X}}
\DeclareMathOperator{\bV}{\mathbf{V}}

\begin{document}

\title[Coherent States and Rational Extensions]{Constructing Coherent
  States for the Rational Extensions of the Harmonic Oscillator
  Potential} \author{Zo\'e McIntyre} \address{Department of Physics,
  McGill U., Montr\'eal QC Canada H3A 2T8}
\email{zoe.mcintyre@mail.mcgill.ca}

\author{Robert Milson}
\address{Dept. of Mathematics and Statistics, Dalhousie U.,
  Halifax NS, Canada B3H 3J5} \email{rmilson@dal.ca}

\begin{abstract}
  Using the formalism of Maya diagrams and ladder operators, we
  describe the algebra of annihilating operators for the class of
  rational extensions of the harmonic oscillators.  This allows us to
  construct the corresponding coherent state in the sense of Barut and
  Girardello.  The resulting time-dependent function is an exact
  solution of the time-dependent Schrodinger equation and a joint
  eigenfunction of the algebra of annihilators.
\end{abstract}

\maketitle

\section*{Introduction}

Supersymmetric quantum mechanics (SUSYQM) has proven to be a key
technique in the construction of exactly solvable potentials and in
the understanding of shape-invariance. The supersymmetric partners of
the harmonic oscillator Hamiltonian are known as rational extensions
because the corresponding potentials consist of the harmonic
oscillator potential plus a rational term that vanishes at infinity.
 
There has been recent interest in rational extensions possessing
ladder operators, which may furnish higher-order analogues of
annihilation operators introduced in the second quantization of the
harmonic oscillator. Ladder operators have applications in the study
of superintegrable systems and rational solutions of Painlev\'e
equations; they also provide a clear avenue for generalizing one of
the defining properties of the canonical coherent states.

In this work, we present a condition, expressed in terms of the Maya
diagram associated with a rational extension, for such a ladder
operator to be designated an annihilation operator. We then construct
the coherent states of the rational extensions as eigenstates of these
annihilation operators.

Note: this is a preliminary version of this article, and will be
revised before publication.
\section*{Preliminaries}
\subsection*{Partitions and Maya diagrams}
A \textit{partition} of a natural number $N$ is a non-increasing
integer sequence $\lambda_1 \geq \lambda_2 \geq \cdots$ such that
\[ |\lambda| := \sum_{i=1}^\infty \lambda_i = N.\] Implicit in this
definition is the assumption that $\lambda_i=0$ for $i$ sufficiently
large.  The length $\ell$ of $\lambda$ is the number of non-zero
elements of the sequence.  It is useful to represent a partition using
a Young diagram
\[ Y_\lambda = \{ (i,j) \colon 1\le i \le \ell,\; 1\le j \le \lambda_i
  \},\] consisting of $\lambda_i$ cells in rows $i=1,\ldots, \ell$.
The hook
\[ \Hk_\lambda(i,j) = \{ (i,k)\in Y_\lambda \colon j\le k \} \cup \{
  (k,j)\in Y_\lambda \colon i\le k \} \] is the set of cells connecting
cell $(i,j)$ to the rim of the diagram.  The hooklength
$\hk_\lambda(i,j)$ is the cardinality of hook $(i,j)\in Y_\lambda$; the
number
\begin{equation}
  \label{eq:dlamdef}
  d_\lambda = \frac{N!}{\prod_{(i,j)\in Y_\lambda}  \hk_\lambda(i,j) }
\end{equation}
counts the number of standard Young tableaux of shape $\lambda$ and
 corresponds to the dimension of an irreducible representation of the symmetric
group $\mathfrak{S}_N$.

\begin{figure}
  \centering
  \begin{tikzpicture}[scale=0.8]
    \draw[step=1cm,black,very thick] (0,0) grid (2,5);
    \draw[step=1cm,black,very thick] (2,0) grid (4,3);
    \draw[step=1cm,black,very thick] (4,0) grid (5,2);
    \path (0.5,0.5) node {9}
    ++(1,0) node {8}
    ++(1,0) node {5}
    ++(1,0) node {4}
    ++(1,0) node {2};
    \path (0.5,1.5) node {8}
    ++(1,0) node {7}
    ++(1,0) node {4}
    ++(1,0) node {3}
    ++(1,0) node {1};
    \path (0.5,2.5) node {7}
    ++(1,0) node {6}
    ++(1,0) node {2}
    ++(1,0) node {1};
    \path (0.5,3.5) node {3} ++(1,0) node {2};
    \path (0.5,4.5) node {2} ++(1,0) node {1};
  \end{tikzpicture}
  \caption{The Young diagram and corresponding hooklengths for the
    partition $(5,5,4,2,2)$.}
\end{figure}

Partitions are closely related to a concept called a \emph{Maya
  diagram}.  We say that a set of integers $M\subset \Z$ is a Maya
diagram if
\[ K_M^+ = \{ m\in M \colon m\geq 0\},\quad K_M^- =\{ m\in \Z\setminus
  M \colon m<0 \} \] are finite sets.  In other words, a Maya diagram
is a subset of $\Z$ that contains a finite number of positive integers
and excludes a finite number of negative integers.  The group $\Z$
acts on $\cM$ by translations, since for $M\in\cM$ and $n\in\Z$,
\begin{equation}
\label{m+n}
M+n=\{m+n:m\in M\}
\end{equation}
is also a Maya diagram.  We will refer to the equivalence class $M/\Z$ of a
partition modulo translations as an \emph{unlabelled} Maya
diagram.  An unlabelled Maya diagram can be represented as a
horizontal sequence of filled $\filledbox$ and empty $\emptybox$
states beginning with an infinite segment of $\filledbox$ and
terminating with an infinite segment of $\emptybox$, where
$\filledbox$ in position $m$ is taken to indicate membership $m\in M$.
A choice of origin serves to convert an unlabelled Maya diagram to a
subset of $\Z$.

There is a natural bijection between the set of partitions and the set
of unlabelled Maya diagrams. For a given partition $\lambda$, the set
\begin{equation}
  \label{eq:Mlam}
  \Mlam = \{ \lambda_i-i \colon i = 1,2,\ldots \}
\end{equation}
is a Maya diagram. 
\begin{proposition}
  \label{prop:Mlambda}
  Every Maya diagram is a translate of a unique $\Mlam$.
\end{proposition}
\begin{proof}
Indeed, let $M\subset \Z$ be a Maya
diagram and $m_1>m_2>\dotsm$ the strictly decreasing sequential
ordering of the elements $m\in M$.  Set
\begin{equation}
  \label{eq:sigmaM}
  \sigma_M = \# \KMpos- \# \KMneg\\
\end{equation}
where $\#$ denotes the cardinality of a finite set. By construction,
\begin{equation}
\label{indextranslate}
\sigma_{M+n}=\sigma_M+n,\quad n\in \Z.
\end{equation}
Hence, there exists a sufficiently large $\ell$ such that
$m_i=-i+\sigma_M$ for all $i> \ell$. 
Next, let
\begin{equation}
  \label{eq:lamfromM}
  \lambda_i = m_i+ i-\sigma_M,\quad i=1,2,\ldots.
\end{equation}
It follows that $\lambda_i=0$ for all $i> \ell$;
i.e., $\lambda$ is a partition. Furthermore, by construction,
\[ M = \Mlam + \sigma_M.\]
\end{proof}
\noindent
We will refer to the integer $\sigma_M$ defined in \eqref{eq:sigmaM}
as the index of $M$, and to the set
\begin{equation}
  \label{eq:KM+-}
   \KM = \KM^+ \cup \KM^-  
\end{equation}
as the index set of $M$.  

The hooklength formula \eqref{eq:dlamdef} can be re-expressed in terms
of a Maya diagram as follows.  Let $\Klam$ denote the index set of
$\Mlam+\ell$, with $k_i= \lambda_i-i+\ell\; (i=1,2,\ldots,\ell)$ 
the decreasing enumeration of $\Klam$.  If $\lambda$ is the partition
corresponding to $M$, then
\[ \prod_{i,j} \hk_\lambda(i,j) = \frac{\prod_i
    k_i!}{\prod_{i<j}(k_i-k_j)}.\]

The bijection between unlabelled Maya diagrams and partitions can be
visualized by representing a filled state with a unit downward arrow
and an empty state with a unit right arrow.  As can be seen in Figure
\ref{fig:bentmaya}, the resulting ``bent'' Maya diagram traces out the
boundary of the Young diagram of the corresponding partition
$\lambda$; see \cite{GGM18} for more details.

\begin{figure}
  \centering
  \begin{tikzpicture}

\draw[step=1cm,gray,dotted,very thin] (0,0) grid (7,8);

\draw[dashed,thick] (0,5) -- +(0,-5) -- +(5,-5);

\draw[line width=2pt,->]    (0,8)  -- ++(0,-1)-- ++(0,-1)-- ++(0,-1);
\draw[line width=2pt,->]    (0,8-3)  -- ++(1,0) -- ++(1,0);
\draw[line width=2pt,->]    (0+2,8-3)  -- ++(0,-1)-- ++(0,-1);
\draw[line width=2pt,->]    (0+2,8-5)  -- ++(1,0) -- ++(1,0);
\draw[line width=2pt,->]    (0+4,8-5)  -- ++(0,-1);
\draw[line width=2pt,->]    (0+4,8-6)  -- ++(1,0 );
\draw[line width=2pt,->]    (5,2)  -- ++(0,-1)-- ++(0,-1);
\draw[line width=2pt,->]    (5,0)  -- ++(2,0);

\draw[thick] (5,0) -- ++(3,0);

\foreach \y  in {5,...,8} {
\draw[dashed,color=gray] (0,\y) -- (5-\y/2,5+\y/2);
};  

\foreach \x  in {0,...,2} {
\draw[dashed,color=gray] (\x,5) -- (2.5+\x/2,7.5-\x/2);
};  

\foreach \y  in {3,...,5} {
\draw[dashed,color=gray] (2,\y) -- (6-\y/2,4+\y/2);
};  

\foreach \x  in {2,...,4} {
\draw[dashed,color=gray] (\x,3) -- (3.5+\x/2,6.5-\x/2);
};  

\foreach \y  in {2,...,3} {
\draw[dashed,color=gray] (4,\y) -- (7-\y/2,3+\y/2);
};  

\foreach \y  in {0,...,2} {
\draw[dashed,color=gray] (5,\y) -- (7.5-\y/2,2.5+\y/2);
};  

\foreach \x  in {5,...,7} {
\draw[dashed,color=gray] (\x,0) -- (5+\x/2,5-\x/2);
};

\draw[fill=black]
    (1.25,8.75) circle (4pt) node[above=4pt] {$-8$}
  ++ (0.5,-0.5)  circle (4pt) node[above=4pt] {$-7$}
  ++ (0.5,-0.5)  circle (4pt) node[above=4pt] {$-6$};

\draw  (1.25,8.75) ++ (1.5,-1.5)
  circle (4pt) node[above=4pt] {$-5$}
 ++ (0.5,-0.5)  circle (4pt) node[above=4pt] {$-4$};

\draw[fill=black]
    (1.25+2.5,8.75-2.5) circle (4pt) node[above=4pt] {$-3$}
  ++ (0.5,-0.5)  circle (4pt) node[above=4pt] {$-2$};

\draw  (1.25+3.5,8.75-3.5)
  circle (4pt) node[above=2pt] {$-1$}
 ++ (0.5,-0.5)  circle (4pt) node[above=2pt] {$0$};

\draw[fill=black]  (1.25+4.5,8.75-4.5)
  circle (4pt) node[above=2pt] {$1$};

\draw  (1.25+5,8.75-5)
  circle (4pt) node[above=2pt] {$2$};

\draw[fill=black]  (1.25+5.5,8.75-5.5)
  circle (4pt) node[above=2pt] {$3$}
++ (0.5,-0.5)  circle (4pt) node[above=2pt] {$4$};

\draw  (1.25+6.5,8.75-6.5)
  circle (4pt) node[above=2pt] {$5$}
 ++ (0.5,-0.5)  circle (4pt) node[above=2pt] {$6$};

\path (3,-0.5) node[font=\itshape,anchor=north,align=center]
{Young Diagram};


\end{tikzpicture}

\caption{The bent Maya diagram with index set $K=\{4,3,1,-1,-4,-5 \}$
  is the rim of the Young diagram of the corresponding partition
  $\lambda=(5,5,4,2,2)$.}
\label{fig:bentmaya}
\end{figure}

Let $\cM$ denote the set of all Maya diagrams.  The flip $\flip_{k}$ at
position $k\in\mathbb{Z}$ is the involution
$\flip_k:\mathcal{M}\rightarrow\mathcal{M}$ defined by
\begin{equation}
\flip_k:M\mapsto \begin{cases}
M\cup\{k\},& k\notin M\\
M\setminus\{k\},&k\in M
\end{cases}.
\end{equation}
\noindent
In the event that $k\notin M$, the flip $\flip_k$ is said to act on $M$
by a state-deleting transformation $\emptybox\rightarrow\filledbox$,
while in the opposite scenario ($k\in M$), it is said to act by a
state-adding transformation $\filledbox\rightarrow\emptybox$.

Let $\cZ$ denote the set of all finite subsets of $\Z$. For a finite
set of integers $K=\{ k_1,\ldots, k_p \}\in \cZ$ we define the
corresponding multi-flip to be the transformation
$\flip_K:\cM\rightarrow\cM$ defined according to
\begin{equation}
\flip_K(M)=(\flip_{k_1}\circ\dotsm\circ \flip_{k_p})(M).
\end{equation}

Observe that multi-flips are also involutions. This means that Maya
diagrams together with multifips have the natural structure of a
complete graph $(\cM,\cZ)$.  The unique edge connecting Maya diagrams
$M_1, M_2$ is the integer set
\begin{equation}
\label{edge}
K=M_1\ominus M_2=M_2\ominus M_1,
\end{equation}
where
\[M_1 \ominus M_2 := (M_1\setminus M_2) \cup (M_2\setminus M_1) \] is
the symmetric difference operation.

Equivalently, we may identify Eq$.$
\eqref{edge} as giving the edge between $M_1$ and $M_2$; i.e., the
unique multiflip satisfying $\flip_K(M_1)=M_2$ and $\flip_K(M_2)=M_1$.

Since $(\cM,\cZ)$ is a complete graph, we can define a bijection
$\cZ\rightarrow \mathcal{M}$ given by $K\mapsto \flip_K(\Mz)$, where
\begin{equation}
  \label{eq:Mzdef}
  \Mz = \{ m \in \Z \colon m<0 \}
\end{equation}
denotes the trivial Maya diagram.  Inversely, for $M\in \cM$, the
unique $K\in \cZ$ such that $M=\flip_K(\Mz)$ is simply the index set
$K=K_M$ defined in \eqref{eq:KM+-}.

\subsection*{Vertex operators and Schur functions}
For $k\in \Nz$, define the ordinary Bell polynomials
$B_k(t_1,\ldots, t_k) \in \Q[t_1,\ldots, t_k]$ as the
coefficients of the power generating function
\begin{equation}
  \label{eq:Bgf}
    \exp\left(\sum_{k=1}^\infty t_k z^k\right)
    = \sum_{k=0}^\infty B_k(t_1,\ldots, t_k) z^k,
\end{equation}
where $ \bt=(t_1,t_2,\ldots)$.  The multinomial formula implies that
\begin{equation}
  \label{eq:Bksum}
  \begin{aligned} B_k(t_1,\ldots, t_k) &= \sum_{ \Vert \mu\Vert=k} \frac{t^{\mu_1}_1}{\mu_1!}
\frac{t^{\mu_2}_2}{\mu_2!}  \cdots \frac{t^{\mu_{k}}_{k}}{\mu_{k}!},\qquad
\Vert \mu \Vert = \mu_1 + 2\mu_2 + \cdots + {k} \mu_{k}\\ &=
\frac{t_1^k}{k!} + \frac{t_1^{k-2}t_2}{(k-2)!} + \cdots + t_{k-1} t_1
+ t_k.
  \end{aligned}
\end{equation}
For a partition $\lambda$ of $N$, define the Schur function
$\Slam(t_1,\ldots, t_N)\in \Q[t_1,\ldots, t_N]$ to be the multivariate
polynomial
 \begin{equation}
   \label{eq:Slamdef} \Slam =
   \det(B_{m_i+j})_{i,j=1}^\ell,
 \end{equation}
 where
 \[ m_i = \lambda_i-i, \] and where $B_k=0$ when $k<0$.  Moreover,
 since
 \[ \partial_{t_i} B_j(t_1,\ldots, t_j) = B_{j-i}(t_1,\ldots,
   t_{j-i}),\quad j\geq i,\] we may re-express \eqref{eq:Slamdef} in
 terms of a Wronskian determinant,
\begin{equation}
  \label{eq:Slamwronsk}
  \Slam = \Wr[B_{m_\ell+\ell},\ldots, B_{m_1+\ell}],
\end{equation}
where the Wronskian is taken with respect to $t_1$.

Let $\bX_m= \bX_m(\bt,\partial_\bt),\; m\in \Z$, be the operators
defined by the generating function
\begin{equation}
  \label{eq:bVdef}
\begin{aligned}
  \bV(\bt,\partial_\bt,z)
  &=\exp\lp \sum_{k=1}^\infty  t_k z^k \rp \exp\lp \sum_{j=1}^\infty
    - j^{-1}\partial_{t_j} z^{-j} \rp \\
  &= \sum_{m=-\infty}^\infty \bX_m(\bt,\partial_\bt)    z^m.
\end{aligned}
\end{equation}
Expanding the above formulas gives
\begin{align}
  \label{eq:Xm+}
  \bX_m &= \sum_{j=0}^\infty B_{j+m}( t_1,\ldots,  t_k) B_j\lp
          \partial_{t_1},\ldots,  j^{-1}
          \partial_{t_j}\rp,\; m\ge 0;\\  
  \label{eq:Xm-}
  \bX_m &= \sum_{j=0}^\infty B_{j}( t_1,\ldots,  t_k) B_{j-m}\lp
          -\partial_{t_1},\ldots, - j^{-1}          \partial_{t_j}\rp,\; m<0.
\end{align}
It can be shown that the above operators obey the 
fundamental relation
\begin{align}
  \label{eq:XmXn}
  &\bX_m \bX_n + \bX_{n-1} \bX_{m+1} = 0.
\end{align}
Despite the fact that the $\bX_m(\bt,\partial_\bt)$ are differential
operators involving infinitely many variables, they have a
well-defined action on polynomials.  In particular, when applied to
Schur functions, they function as multi-variable raising operators.
\begin{proposition}
  \label{prop:SlamXl}
  For every partition $\lambda$ of length $\ell$, we have
  \begin{equation}
    \label{eq:SlamXlam}
    \Slam = \bX_{\lambda_1} \cdots \bX_{\lambda_{\ell}} 1,
  \end{equation}
  where $1$ is the Schur function corresponding to the trivial partition.
\end{proposition}
\noindent
The proof of \eqref{eq:XmXn}--\eqref{eq:SlamXlam} can be found in
\cite[Appendix A]{NY99}. As an immediate corollary we obtain the
following.
\begin{proposition}\label{prop:XmSlam}
  Let  $\lambda$ be a partion, $\Mlam\subset \Z$ the corresponding Maya diagram \eqref{eq:Mlam}, and
  \begin{equation}
    \label{eq:Jlam}
    \Jlam = \Z\setminus \Mlam
  \end{equation}
  the corresponding integer complement.  Then, for every $m\in\Z$ we
  have
  \begin{equation}
    \label{eq:Xmact}
    \bX_m  \Slam =
    \begin{cases} 
      (-1)^{\# \{ k\in \Mlam : k> m \}} S_{m\triangleright\lambda} &
      \text{ if } m\in \Jlam\\
      0 & \text{ if } m\in \Mlam
    \end{cases}.
  \end{equation}
\end{proposition}

By construction, the action of $\bV(\bt,z)$ on a polynomial
$P(\bt)\in \C[t_1,\ldots, t_n]$ can be given as
{\small
\begin{equation}
  \label{eq:bVP}
  \bV(\bt,z) P(\bt) = \exp\lp \sum_{k=1}^\infty t_k z^k \rp P\!\lp t_1-
  z^{-1}, t_2 - \frac{z^{-2}}{2}, \ldots, t_n - \frac{z^{-n}}{n}\rp.
\end{equation}}
Proposition~\ref{prop:XmSlam} allows the action of $\bV(\bt,z)$ on a Schur
polynomial to be conveniently written in terms of the ``insertion''
procedure:
\begin{theorem}
  Let $\lambda$ be a partition. With the above notation, we have
  \begin{equation}
    \label{eq:Vlam}
    \bV(\bt,z) \Slam(\bt) = \sum_{m\in\Jlam}
    (-1)^{\# \{ k\in \Mlam : k> m \}} S_{m\triangleright\lambda}(\bt)  z^m.
  \end{equation}
\end{theorem}
\subsection*{Hermite polynomials}
Hermite polynomials $H_n(x),\; n=0,1,\ldots$, are univariate
polynomials defined by
\begin{equation}
  \label{eq:HnRodrigues}
  H_n(x)=(-1)^ne^{x^2}\frac{d^n}{dx^n}e^{-x^2}, n=0,1,2,\dots\:.
\end{equation}
The $H_n(x)$ are known as classical orthogonal polynomials because they
satisfy a second-order eigenvalue equation
\begin{equation}
  \label{eq:hermde}
  y''-2xy' = 2n y,\quad y= H_n(x),
\end{equation}
as well as a 3-term recurrence relation
\begin{equation}
  \label{eq:h3term}
  H_{n+1}(x) =  2x H_n(x) - 2n H_{n-1}(x),
\end{equation}
which together imply the following orthogonality relation:
\begin{equation}
  \label{eq:hortho}
  \int_{\R} H_m(x) H_n(x) e^{-x^2} dx = \sqrt{\pi}\, 2^n
  n! \delta_{n,m}.
\end{equation}

The generating function for the Hermite polynomials is
\begin{align}
  \label{eq:hermgf}
    e^{xz - \tfrac14 z^2}
    &= \sum_{n=0}^\infty H_n(x) \frac{z^n}{2^nn!}.
\end{align}
To prove this statement \eqref{eq:hermgf}, write
\[ e^{xz - \tfrac14 z^2} = \sum_{n=0}^\infty f_n(x)
  \frac{z^n}{2^nn!}\] and observe that
\[
  (\partial_x+ 2\partial_z) \lp e^{xz - \tfrac14z^2- x^2}\rp =
  (\partial_x+ 2\partial_z) e^{-\tfrac14(2x-z)^2} =  0.
\]
Hence,
\[
  \partial_x e^{xz - \tfrac14z^2-x^2}
  = -2\partial_z e^{xz - \tfrac14z^2-x^2}
  = - \sum_{n=0}^\infty f_{n+1}(x) e^{-x^2} \frac{z^n}{2^n n!}.
\]
It follows that
\[ \frac{d}{dx} \lp f_n(x) e^{-x^2}\rp = - f_{n+1}(x)e^{-x^2},\] so since
$f_0(x)=1$, we conclude that $f_n(x) = H_n(x)$ for all
$n=0,1,\dots$.

Comparison of \eqref{eq:hermgf} with \eqref{eq:Bgf} shows that the
Hermite polynomials are specializations of Bell polynomials:
\[
  \begin{aligned}
    H_n(x) &= n! 2^n B_n(x,-\tfrac14,0,\ldots)   .\\
  \end{aligned}
\]
Applying \eqref{eq:Bksum} then gives the well-known formula
\[ H_n(x) = \sum_{j=0}^{\lfloor n/2\rfloor} (-1)^j \frac{n!}{(n-2j)!
    j!} (2x)^{n-2j}.\] 

In the sequel, we will also make use of the
conjugate Hermite polynomials.  These can be defined using the
following equivalent relations:
\begin{equation}
  \label{eq:tHndef}
 \begin{aligned}
   \tH_n(x) &= n! 2^n B_n(x,\tfrac14,0,\ldots),\quad   n=0,1,2,\ldots\\
   &=e^{-x^2}\frac{d^n}{dx^n}e^{x^2}\\
   &= i^{-n} H_n(i x).
 \end{aligned}
\end{equation}

\subsection*{The  harmonic oscillator}
Write $p= i \partial_x$, so that
\[ T(x,\partial_x) = p^2+x^2=- \partial_x^2 + x^2 \] is the 
Hamiltonian of the quantum harmonic oscillator.  We say that a function
$\psi(z)$ is quasi-rational if its log-derivative, $\psi'(z)/\psi(z)$,
is a rational function.  The quasi-rational eigenfunctions of $T$ are
\begin{equation}
  \label{eq:psindef}
\psi_n(x)=\begin{cases}
e^{-\tfrac{x^2}{2}}H_n(x), & n\geq 0\\
e^{\tfrac{x^2}{2}}\tH_{-n-1}(x), & n<0
\end{cases},
\end{equation}
where $H_n(x)$ and  $\tH_n(x)$ are the Hermite and conjugate Hermite
polynomials defined in \eqref{eq:HnRodrigues} and \eqref{eq:tHndef}. 
The eigenfunctions $\psi_n,\; n\ge 0$, represent the bound states of
the harmonic oscillator, while those with $n<0$ do not satisfy the boundary
conditions at $\pm\infty$ and instead represent virtual states. The corresponding eigenvalue relation is
\begin{equation}
  \label{eq:Tpsin}
  T\psi_n = (2n+1) \psi_n,\quad n\in \Z.
\end{equation}

It will be instructive to prove \eqref{eq:Tpsin} using generating
functions.  Multiplication of \eqref{eq:hermgf} by $e^{-x^2/2}$ yields
\begin{equation}
  \label{eq:Psi0def}
  \Psi_0(x,z) = e^{-\tfrac12(x-z)^2+\tfrac14z^2},
\end{equation}
which serves as
the generating function for the bound states of the harmonic
oscillator:
\begin{equation}
  \label{eq:psingf}
  \Psi_0(x,z) = \sum_{n=0}^\infty \psi_n(x) \frac{z^n}{2^n n!}.
\end{equation}

By a direct calculation, we have
\begin{equation}
  \label{eq:TPsi0}
   T(x,\partial_x) \Psi_0(x,z) = (2 z \partial_z+1)\Psi_0(x,z). 
\end{equation}
Applying the above relation to \eqref{eq:psingf} and comparing the
coefficients of the resulting power series then returns \eqref{eq:Tpsin}.

The classical ladder operators
\begin{equation}
  \label{eq:L+idef}
  \begin{aligned}
    L_{\mp}(x,\partial_x) := \partial_x \pm x
  \end{aligned}
\end{equation}
satisfy the intertwining relations
\[ T L_- = L_- (T-2),\qquad T L_+ = L_+(T+2).\]
As a consequence, we have the following lowering and raising relations
for the bound states:
\begin{equation}
  \label{eq:lrrel}
  \begin{aligned}
    L_- \psi_n &= 2n \psi_{n-1},\quad n=0,1,2,\ldots\\
    L_+ \psi_n &=  \psi_{n+1}.
  \end{aligned}
\end{equation}
Relations \eqref{eq:lrrel} can also be established using generating
functions; it suffices to observe that
\begin{align}
  \label{eq:L-Psi0}
  L_-(x,\partial_x) \Psi_0(x,z)
  &= z \Psi_0(x,z) = \sum_{n=0}^\infty 2n \psi_{n-1}(x)
    \frac{z^n}{2^n n!};\\
  L_+(x,\partial_x) \Psi_0(x,z) &= 2\partial_z \Psi_0(x,z) = \sum_{n=0}^\infty \psi_{n+1}(x)
                                  \frac{z^n}{2^n n!}.
\end{align}

\subsection*{The  canonical coherent state}

The canonical coherent state
\begin{equation}
  \label{eq:phidef}
 \Phi_0(x,t;\alpha) := \exp\lp -it + \tfrac12 x^2 - (x-\tfrac12 e^{-2
    it}\alpha)^2\rp   
\end{equation}
is a time-dependent eigenfunction of the lowering operator
\begin{equation}
  \label{eq:L-phi}
  L_-(x,\partial_x) \Phi_0(x,t;\alpha) = \alpha e^{-2it}
  \Phi_0(x,t;\alpha),
\end{equation}
as well as an exact solution to the time-dependent Schrodinger equation:
\begin{equation}
  \label{eq:IDtTphi}
  (-i \partial_t+ T(x,\partial_x)) \Phi_0(x,t;\alpha) = 0.
\end{equation}

Observe that 
\begin{equation}
  \label{eq:Phi0Psi0}
\Phi_0(x,t;\alpha) = e^{-it} \Psi_0(x,\alpha e^{-2it}).
\end{equation}
Hence, the canonical coherent state may also be regarded a generating
function for the bound states of the harmonic oscillator.  Indeed, the
change of variables \eqref{eq:Phi0Psi0} transforms the eigenvalue
relation \eqref{eq:L-phi} into \eqref{eq:L-Psi0}, and the
TDSE \eqref{eq:IDtTphi} into relation \eqref{eq:TPsi0}.

\subsection*{Pseudo-Wronskians}
Let $M\in \cM$ be a Maya diagram. Also, let  $0>s_1\ge\cdots \ge s_p$ be the
elements of $\KMneg$ and $0\le t_1\le \cdots\le t_q$ the elements of
$\KMpos$ arranged in the indicated order.   Define 
\begin{equation}\label{eq:pWdef1}
  \HM = e^{-px^2}\Wr[ e^{x^2} \tH_{-s_1-1},\ldots, e^{x^2}
  \tH_{-s_p-1}, H_{t_1},\ldots H_{t_q} ],
\end{equation}
where $H_n(x)$ and $\tH_n(x),\; n=0,1,2,\ldots$, are the classical Hermite
and conjugate Hermite polynomials, and where $\Wr$ denotes the
Wronskian determinant of the indicated functions.  The polynomial
nature of $\HM$ becomes evident once we represent it as the following
\textit{pseudo-Wronskian} determinant \cite{GGM18}:
\begin{equation}\label{eq:pWdef2}
  \HM =
  \begin{vmatrix}
    \tH_{s_1} & \tH_{s_1+1} & \ldots & \tH_{s_1+p+q-1}\\
    \vdots & \vdots & \ddots & \vdots\\
    \tH_{s_p} & \tH_{s_p+1} & \ldots & \tH_{s_p+p+q-1}\\
    H_{t_q} & D_x H_{t_q} & \ldots & D_x^{p+q-1}H_{t_q}\\
    \vdots & \vdots & \ddots & \vdots\\
    H_{t_1} & D_x H_{t_1} & \ldots & D_x^{p+q-1}H_{t_1}
  \end{vmatrix}.
\end{equation}
The proof of \eqref{eq:pWdef2} can be found in \cite{GGM18}.  The same
article also showed that the normalized polynomial
\begin{equation}
  \label{eq:hHdef}
\hH_M=\frac{(-1)^{pq}H_M}{\prod_{i<j}2(s_j-s_i)\prod_{i<j}2(t_j-t_i)}
\end{equation}
is translation invariant:
\begin{equation}
  \label{eq:HM+n}
  \hH_M=\hH_{M+n},\quad n\in\Z.
\end{equation}

 It follows that the normalized Hermite
pseudo-Wronskian \eqref{eq:hHdef} has the following expression in
terms of Schur functions:
\begin{equation}
  \label{eq:HMSlam}
  \hH_M(x) =  \frac{2^{N}N!}{d_\lambda} \, \Slam(x,-\tfrac14,0,\ldots ),
\end{equation}
where $N=|\lambda|$ and where $d_\lambda$ is the combinatorial factor
defined in \eqref{eq:dlamdef}.

\subsection*{ Rational Extensions}
In this section, we describe the correspondence between Maya diagrams
and rational extensions of the harmonic oscillator.

Let $M\in \cM$ be a Maya diagram. The
pseudo-Wronskian defined in \eqref{eq:pWdef1} can now be expressed
simply as
\begin{equation}
  \label{eq:HM}
  H_M(x)=e^{\sigma_M\tfrac{x^2}{2}}\Wr[\psi_{k_1}(x),\dots,\psi_{k_p}(x)],
\end{equation}
where $\psi_n(x),\; n\in \Z$, are the quasi-rational eigenfunctions
\eqref{eq:psindef}, and where $\sigma_M$ is the index defined in
\eqref{eq:sigmaM}.  The potential
\begin{align}
\label{rationalext}
  U_M(x)
  &=x^2-2\frac{d^2}{dx^2}\log\Wr[\psi_{k_1},\dots,\psi_{k_p}]\\ \nonumber
  &=x^2+2\left(\frac{H'_M}{H_M}\right)^2-\frac{2H_M^{''}}{H_M}-2\sigma_M   
\end{align}
is a rational extension of the harmonic oscillator potential,
so called because the terms following the $x^2$ in~\eqref{rationalext}
are all rational.

The corresponding Hamiltonian operator
\begin{equation}
  T_M:=-\frac{d^2}{dx^2}+U_M
\end{equation}
is  exactly solvable \cite{GGM13} with eigenfunctions
\begin{align}
\label{eigenstates}
  \psi_{M,m}=
  e^{\epsilon_M(m)\tfrac{x^2}{2}}\frac{\hH_{\flip_m(M)}}{\hH_M},\quad
  \epsilon_M(m) =
  \begin{cases}
    -1 & \text{ if } m\notin M\\
    +1 & \text{ if } m\in M\\
  \end{cases},\quad m\in \Z
\end{align}
and eigenvalues
\begin{equation}
  \label{eq:TMpsiMn}
  T_M\psi_{M,m}=(2m+1)\psi_{M,m},\quad m\in \Z.
\end{equation}

(In \eqref{eigenstates}, $\hH_M$ is the normalized pseudo-Wronskian defined in
\eqref{eq:hHdef}.) Relation \eqref{eq:HM+n} together with Eq.~\eqref{indextranslate}
implies that $T_M$ and the corresponding eigenfunctions are
translation covariant:
\begin{equation}
  \begin{aligned}
    T_{M+n}&=T_M+2n,\quad n\in\Z\\
    \psi_{M+n,m+n}&= \psi_{M,n}.
  \end{aligned}
\end{equation}

\smallskip As regards integrability, it should be noted that by the
Krein-Adler theorem \cite{GGM13}, $H_M$ has no real zeros if and only
if all finite $\filledbox$ segments of $M$ have even size. It is
precisely for such $M$ that $T_M$ corresponds to a self-adjoint
operator and that the eigenfunctions $\psi_{M,m}$ are
square-integrable for $m\in\Z\setminus M$. The set
\begin{equation}
  \label{eq:IMdef}
  I_M := \Z\setminus M = \Jlam + \sigma_M
\end{equation}
of empty boxes of $M$ then serves as the index set for the bound states.

The generating function for the bound states of a rational extension
can be given as follows:  
\begin{proposition}
For a partition $\lambda$, define
\begin{equation}
  \label{eq:Psilamdef}
  \Psilam(x,z) = \frac{\Slam\lp x-z^{-1}, -\tfrac14
    -\tfrac12z^{-2},-\tfrac13 z^{-3},\ldots\rp}{\Slam(x,-\tfrac14,0,\ldots)} \Psi_0(x,z).
\end{equation}
Let $M\in \cM$ be a Maya diagram and $\lambda$ the corresponding
partition as per Proposition \ref{prop:Mlambda}. Then
\begin{equation}
  \label{eq:Psilamgf}
  \Psilam(x,z) = \sum_{m\in I_M}  \psi_{M,m}(x)
  \frac{\prod_{i=1}^\ell (m-m_i)}{(m-\sigma_M+\ell)!}
  \lp\frac{z}{2}\rp^{m-\sigma_M}, 
\end{equation}
where $m_1>m_2>\cdots$ is the decreasing enumeration of $M$.
\end{proposition}
\begin{proof}
  This follows from \eqref{eq:bVdef}, \eqref{eq:Xmact} and \eqref{eq:bVP}.
\end{proof}

\noindent
Observe that if $M=\Mz$ is the trivial Maya diagram, then
\eqref{eq:Psilamgf} reduces to the classical generating function shown
in \eqref{eq:psingf}.

\vspace{2ex}

\subsection*{Ladder Operators}
In this section, we introduce ladder operators for the rational
extensions $T_M,\; M\in \cM$, defined above.  

Intertwining relations
have their origins in supersymmetric quantum mechanics (SUSYQM).  For
differential operators $A$, $T_1$, $T_2$, we say that $A$ intertwines
$T_1$ and $T_2$ if
\begin{equation}
  \label{eq:intertwiner}
  AT_1=T_2A.
\end{equation}
We will refer to $A$ as a ladder operator if $T_2=T_1+\lambda$ for
some constant $\lambda$.  As a direct consequence of
\eqref{eq:intertwiner}, $A$ maps  eigenfunctions of $T_1$ to
eigenfunctions of $T_2$, possibly annihilating finitely many eigenfunctions.
This means that if $A$ is a ladder operator, then $A$ has  a
well-defined raising or lowering action on the states of $T_1$.

Within the class of rational extensions, the intertwiners take the form
\begin{equation}
  \label{eq:AMKdef}
  A_{M,K}[y]=\frac{\Wr[\psi_{M,m_1},\dots,\psi_{M,m_p},y]}{\Wr[\psi_{M,m_1},\dots,\psi_{M,m_p}]},
\end{equation}
where $M\in \cM$ is a Maya diagram, $K\in \cZ$ is a finite set
of integers, and where the $\Psi_{M,m}$ are the quasi-rational
eigenfunctions of $T_M$ defined in \eqref{eigenstates}.  By
construction, $A_{M,K}$ is a monic differential operator of order $p$
that intertwines $T_M$ and $T_{\flip_K(M)}$.  

Intertwiners
$A_{M_1,K_1}$ and $A_{M_2,K_2}$ such that $M_2= \flip_{K_1}(M_1)$ can be
composed according to 
\begin{equation}
\label{comp}
A_{M_2,K_2}\circ A_{M_1,K_1}=A_{M_1,K_{1}\ominus K_2}\circ p_{K_1,K_2}(T_M),
\end{equation}
where
\[
  \begin{aligned}
    p_{K_1,K_2}(x) &= \prod_{k\in K_1\cap K_2} (2k+1-x).
\end{aligned}
\]
Since $T_{M+n}=T_M+2n$, the above intertwiners
are also translation invariant:
\begin{equation}
\label{translation}
A_{M+n,K+n}=A_{M,K},\quad n\in\Z.
\end{equation}
For additional details, see Section 4 of \cite{GGMM20}.

Now fix an integer $n\in \Z$, and let
\begin{align}
  \label{eq:Kndef}
  \KMn &:= (M+n) \ominus M,\\
  \LMn &:= A_{M,\KMn}.
\end{align}
Theorem 4.1 of \cite{GGMM20} implies that, in
this particular case, the intertwining relation takes the form
\begin{equation}
\label{ladder}
\LMn T_M=(T_M+2n)\LMn.
\end{equation}
We will refer to such an $\LMn$ as a \textit{ladder operator} for the rational
extension $T_M$. 
The action of ladder operators on states is that of a lowering
or raising operator according to
\[ L_n[\psi_{M,k}] = C_{M,n,k} \psi_{M,k-n},\quad k\notin M, \] where
$C_{M,n,k}$ is zero if $\psi_{M,k-n}$ is not a bound state, i.e., if
$k-n\in M$. Otherwise, $C_{M,n,k}$ is a rational number whose explicit
form is given in \cite{GGMM20}.

\subsection*{The annihilator algebra}
In this section we describe the annihilation operators of a rational
extension.  The situation is more complicated than in the canonical
case, where the annihilation operators are generated by
$L_-=\partial_x+x$. For a non-trivial rational
extension, the annihilation operators form a non-trivial ring of
commuting operators with a structure determined by the combinatorics
of the corresponding Maya diagram, as we will now show.

For $q\in \N$, we say that a Maya diagram $M\in \cM$ is a $q$-core if
$M\subset M+q$. In such a case, the symmetric difference
$(M+q)\ominus M$ is nothing but the set difference $(M+q)\setminus M$.
It follows that the kernel of the ladder  operator $\LM^{(q)}$ is
spanned by $\psi_{M,m},\; m\in (M+q)\setminus M$.

We  say that $q\in \N$ is a critical degree of a Maya diagram
$M\in \cM$ if $M$ is a $q$-core.  Observe that if $q$ is a critical
degree of $M$, then $q$ is a critical degree of $M+n$ for every
$n\in \N$.  Thus, the $q$-core property is an attribute of an
unlabelled Maya diagram.  The set of unlabelled Maya diagrams is
naturally bijective to the set of partitions, and so we use $\Dlam$,
where $\lambda$ is the partition defined in \eqref{eq:lamfromM}, to
denote the set of all critical degrees of the unlabelled Maya diagram
$M/\Z$. This definition is consistent with the definition of
the $q$-core partition used in combinatorics; see \cite{Mc98} for more
details.

A $q\in \N$ fails to be in $\Dlam$ if and only if there exists an
occupied position on a Maya diagram $m\in M$ and an unoccupied
position $k\in I_M$ such that $k=m-q$.  The
smallest empty position on a Maya diagram occurs at position
$m_{\ell+1}+1 = \sigma_M-\ell$, while the largest occupied position occurs at
$m_1= \lambda_1-1+\sigma_M$.  It then follows that
\begin{equation}
  \label{eq:qcdef}
  q_c := m_1 - (\sigma_M-\ell)+1 =  \lambda_1+\ell
\end{equation}
is a threshold critical degree, in the sense that $q\in \Dlam$ for all
$ q\ge q_c$ and $q_c-1\notin \Dlam$.  See Figure
\ref{fig:lowering} for an example.

Let $M\in \cM$ be a Maya diagram and $\lambda$ the corresponding
partition. Set $\Rlam =\lspan \{ z^q \colon q\in \Dlam\}$, and observe
that if $q_1, q_2\in \Dlam$, then $q_1+q_2\in \Dlam$ also.  It follows
that $\Rlam$ is closed with respect to multiplication; i.e$.$ $\Rlam$ is
an operator algebra.  Also note that composition of annihilation operators on
the left of \eqref{eq:LqPsiz} is equivalent to multiplication of
eigenvalues on the right; it follows that $\Rlam$ is isomorphic to
the ring of annihilation operators associated with the rational
extension $T_M$.  

Relations \eqref{eq:Psilamgf} and \eqref{eq:LqPsiz}
entail the following action of the annihilators on the bound states:
\begin{equation}
  \label{eq:Lqpsim}
  \LM^{(q)}(x,\partial_x) \psi_{M,m}(x) = 2^q
  \gamma_{M}^{(q)}(m)\psi_{M,m-q}(x),
\end{equation}
where $m\in I_M, q\in \Dlam$, and where
\[ \gamma_M^{(q)}(x) = \prod_{k\in (M+q)\setminus M} (x-k).\] Observe
that $m-q\notin I_M$ if and only if $m\in M+q$. Hence,
$\gamma_M^{(q)}(m)=0$ when $\psi_{M,m}(x)$ is a
well-defined bound state, but $\psi_{M,m-q}(x)$ fails to satisfy the
boundary conditions at $\pm \infty$.

\subsection*{A Coherent State for Rational Extensions}
We are now ready to present the construction wherein one may
generalize the notion of a coherent state to any rational extension
$T_M$ of the harmonic oscillator potential.  We proceed, as in the
canonical case, by constructing the coherent state in terms of the
generating function. In \cite{KM20}, it was shown that, in terms of the
generating function $\Psilam(x,z)$, the eigenvalue relation
\eqref{eq:TMpsiMn} is equivalent to
\begin{equation}
  \label{eq:TMpsilam}
  T_M(x,\partial_x) \Psilam(x,z) = (z\partial_z + 1 + 2\sigma_M) \Psilam(x,z).
\end{equation}
Using the same change of variables as in \eqref{eq:Phi0Psi0}, let us
therefore set
\begin{equation}
  \label{eq:Philamdef}
 \Philam(x,z) = e^{-(1+2\sigma_M)it}\Psilam(x,\alpha e^{-2 it }).
\end{equation}
Then by construction, $\Philam(x,t)$ is an exact solution of the
time-dependent Schr\"odinger equation corresponding to the rational
extension $T_M$:
\begin{equation}
  \label{eq:PhilamTM}
    i \partial_t \Philam(x,t) = T_M(x,\partial_x) \Philam(x,t) .
\end{equation}

By Theorem 6.1 of \cite{KM20}, for every
critical degree $q\in \Dlam$, we have
\begin{equation}
  \label{eq:LqPsiz}
  \LM^{(q)}(x,\partial_x) \Psilam(x,z) = z^q \Psilam(x,z).
\end{equation}
In other words, the generating function is a joint eigenfunction of
the annihilator algebra.  From here, it isn't difficult to modify
\eqref{eq:LqPsiz} to obtain an eigenrelation for the coherent state
$\Philam(x,z)$ defined in \eqref{eq:Philamdef}.  Indeed, by a
straight-forward change of variables, we have that
\begin{equation}
  \label{eq:LqPhial}
  \LM^{(q)}(x,\partial_x) \Philam(x,t;\alpha) = \alpha^q e^{-2iqt} \Philam(x,z).
\end{equation}
Hence, $\Philam(x,t;\alpha)$ is a joint eigenfunction of the
annihilator algebra and satisfies the definition of a
coherent state in the sense of Barut-Girardello \cite{BG71}.

\subsection*{Example.} As an example, we construct the coherent state
corresponding to the index set $K=\{2,3\}$.  The corresponding Maya diagram, partition, and index are
\[ M= f_K(\Mz) = \Mz \cup \{ 2,3 \},\quad \lambda=(2,2),\quad
  \sigma_M=2,\]
while the corresponding rational extension is
\[ T_M(x,\partial_x) =-\partial_x^2+ \left(x^2+4+\frac{32 x^2}{4
      x^4+3}-\frac{384 x^2}{\left(4 x^4+3\right)^2}\right). \]

    Table \ref{tab:HMm} shows
the numerator polynomials of the first few bound states of $T_M$. These bound states are indexed by
\[ I_M  = \{ 0,1,4,5,6,\ldots \} , \]
and explicitly, the bound state with eigenvalue $2m+1,\; m\in I_M$, can be written as
\[ \psi_{M,m} =  e^{-\tfrac12x^2} \frac{H_{M,m}(x)}{4x^4+3}
,\quad m\in \Jlam+2,\] where
\[ H_{M,m} = \frac{\Wr(2x^2-1,2x^3-3,H_m)}{4 (m-2)(m-3)(4
	x^4+3)},\quad m\ge 0,\; m\neq 2,3. \]

\begin{table}
  \centering
  \begin{tabular}{r|l}
    $m$ & $H_{M,m}(x)$\\ \hline
    0 & $x^2+\frac{1}{2}$\\
    1 & $2 x^3+3 x$\\
    4 & $16 x^6+24 x^4+36 x^2-18$\\
    5 & $32 x^7+16 x^5+40 x^3-60 x$\\
    6 & $64   x^8-64 x^6-240 x^2+60$
  \end{tabular}
  \caption{Bound states of the rational extension corresponding to the
    index set $K=\{ 2,3\}$.}
  \label{tab:HMm}
\end{table}

In this case, the Schur function is
\[ \Psilam(t_1,t_2,t_3) = \frac{t_1^4}{12}+t_2^2-t_1 t_3.\] Using
\eqref{eq:Psilamdef}, the generating function for the bound states is
therefore
\[ \Psilam(x,z) = \lp 1- \frac{16x^3}{4x^4+3} z^{-1}+
  \frac{12(2x^2+1)}{4x^4+3} z^{-2} \rp
  e^{-\tfrac12(x-z)^2+\tfrac14z^2}.\]

The set of critical degrees, meanwhile, is $\Dlam = \{ 4,5,\ldots \}$. Note that there are no critical degrees below the threshold $q_c=2+2=4$.
Figure \ref{fig:lowering} illustrates the fact that $q=4$ is a
critical degree and that $q=3$ fails to be a critical degree since
$0+3\in M$ but $0\notin M$.

The coherent state wave function
\[ \Philam(x,t;\alpha) = e^{-5 it} \Psilam(x,\alpha e^{-2it}) \] is an
exact solution of the corresponding time-dependent Schr\"odinger
equation \eqref{eq:PhilamTM}.
The annihilators 
\begin{align*}
  \LM^{(4)}  &= A_{M,\{0,1,6,7\}},\\
  \LM^{(5)}  &= A_{M,\{0,1,4,7,8\}},\\
  \LM^{(6)}  &= A_{M,\{0,1,4,5,8,9\}},\\
  \LM^{(7)}  &= A_{M,\{0,1,4,5,6,9,10\}}
\end{align*}
are commuting differential operators that generate the annihilator
algebra of this rational extension.  In each case, direct calculation
shows that $\Philam(x,t;\alpha)$ is an eigenfunction of $\LM^{(q)}$
with eigenvalue $\alpha^q e^{-2qi t}$.

\begin{figure}[h]
	\centering
	\begin{tikzpicture}[scale=0.6]
	
	\path [fill] (0.5,2.5) node {\huge ...}
	++(1,0) circle (5pt) ++(1,0) circle (5pt)  ++(1,0) circle (5pt)
	++(3,0) circle (5pt) ++(1,0) circle (5pt)
	++(9,0) node {\huge ...} +(1,0);

	\foreach \x in {-3,...,11} 	\draw (\x+4.5,3.5)  node {$\x$};
	\path (6.5,0.5) node {} ++ (1,0) node {}
	++ (2,0) node {}++ (2,0) node {}++ (3,0) node {}
	;
	
	\draw  (1,0) grid +(15 ,3);
	
	\path [fill] (0.5,1.5) node {\huge ...}
	++(1,0) circle (5pt)
        ++(1,0) circle (5pt)
	++(1,0) circle (5pt)
	++(1,0) circle (5pt)
	++(1,0) circle (5pt)
	++(1,0) circle (5pt)
	++(3,0) circle (5pt) ++(1,0) circle (5pt)
	++(6,0) node {\huge ...} ;
	\path [fill] (0.5,0.5) node {\huge ...}
	++(1,0) circle (5pt) ++(1,0) circle (5pt)  ++(1,0) circle (5pt)
	++(1,0) circle (5pt)
	++(1,0) circle (5pt)
	++(1,0) circle (5pt)
	++(1,0) circle (5pt)
	++(3,0) circle (5pt) ++(1,0) circle (5pt)
	++(5,0) node {\huge ...} +(1,0);
	
	\draw[line width=1pt] (4,1) -- ++ (0,1.5);
	
	\end{tikzpicture}
	\captionsetup{font=footnotesize}
	\caption{Top: The Maya diagram $M$ corresponding to index set
          $K=\{2,3\}$.  The corresponding partition and index are
          $\lambda=(2,2)$ and $\sigma_M=2$, respectively, while the threshold
          critical degree is $q_c = 4$. Middle: $M+3$. Bottom:
          $M+4$. Note that $4$ is a critical degree since
          $M\subset M+4$. However, $3$ fails to be a critical degree
          since $3\in M$ but $3\notin M+3$. }
        \label{fig:lowering}
\end{figure}
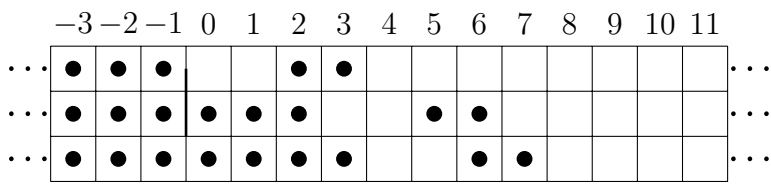

\end{document}